\documentclass{article}

\usepackage{graphicx} 
\usepackage{amsfonts}
\usepackage{amsthm}
\usepackage{amsmath}
\usepackage{subcaption}
\usepackage{tikz}
\usetikzlibrary{positioning}
\usepackage[showdeletions]{color-edits}

\usepackage{dsfont}

\addauthor{ac}{red}
\addauthor{so}{blue}

\newtheorem{theorem}{Theorem}
\newtheorem{lemma}{Lemma}

\theoremstyle{definition}
\newtheorem{example}{Example}[section]

\newcommand{\E}{\mathbb{E}}
\newcommand{\PP}{\mathbb{P}}
\newcommand{\xhdr}[1]{\paragraph{\bf #1.}}

\usepackage[letterpaper]{geometry}

\usepackage{hyperref}
\usepackage{cleveref}

\author{Andrés Cristi$^1$ \and Sigal Oren$^2$}
  
\date{%
$^1$ Center for Mathematical Modeling, Universidad de Chile, Chile\\
$^2$ Ben-Gurion University of the Negev, Israel
}

\title{Planning Against a Prophet: A Graph-Theoretic Framework for Making Sequential Decisions}

\begin{document}

\maketitle
\begin{abstract}
We devise a general graph-theoretic framework for studying prophet inequalities. In this framework, an agent traverses a directed acyclic graph from a starting node $s$ to a target node $t$. Each edge has a value that is sampled from a known distribution. When the agent reaches a node $v$ it observes the realized values of all the outgoing edges from $v$. The agent's objective is to maximize the expected total value of the path it takes. As in prophet inequalities, we compare the agent's performance against a prophet who observes all the realizations of the edges' values ahead of time. Our analysis reveals that this ratio highly depends on the number of paths $k$ required to cover all the nodes in the graph. In particular, we provide an algorithm that guarantees a prophet inequality ratio of $\frac{1}{2k}$ and show an upper bound of $\frac{1}{k+1}$. 

Our framework captures planning problems in which there is uncertainty regarding the costs/benefits of each action. In particular, it captures an over-time variant of the classic prophet inequality in which a seller leases a durable item, such as an apartment, for $n$ time units. Each period a lessee appears and may have a different value for each lease term. We obtain a tight bound of $1/2$ for this variant. To make this framework even more expressive, we further generalize it to accommodate correlations between edges originating from the same node and allow for additional constraints on the edges the agent can take. The generalized framework captures many well-studied prophet inequality problems, including $d$-dimensional  matching, $k$-prophet inequality, and more. 

\end{abstract}

\newpage

\section{Introduction}

The classic prophet inequality problem \cite{krengel1978semiamarts} considers a decision maker (DM) who is interviewing candidates in order to select the candidate with the highest value. The candidates arrive sequentially, and initially, the DM only knows the distributions from which their values are drawn. At each step $i$, the DM observes the value of candidate $i$ and needs to make an irrevocable decision whether to hire her or not. Once a candidate is hired the process ends.
This line of research focuses on the competitive ratio, which compares the value of the candidate that the DM hires to the maximum realized value of a candidate. The latter is the benchmark achieved by a prophet who observes all the realizations of the candidates' values ahead of time. An algorithm that can guarantee an expected value that is a $c$ fraction of the expected value of the prophet is a $c$-competitive algorithm. Classic results by Krengel and Sucheston~\cite{krengel1978semiamarts} and Samuel-Cahn~\cite{samuel1984comparison} show that various simple algorithms are $1/2$-competitive, and this is the best competitive ratio attainable.

Since its inception, the prophet inequality problem has inspired a magnificent body of work spanning various academic fields (See \cite{lucier2017economic} for a survey of some of the recent work in the computer science community). From a broader perspective, these problems often involve a stochastic setting where a decision-maker (DM), armed only with distributional information, makes sequential decisions to compete against an all-knowing prophet. Many planning problems can be modeled as selecting a path in a graph where nodes represent different states, and the previously mentioned broader perspective can also be applied to this type of planning problem. In this paper, we utilize this connection to develop a graph theoretic framework for prophet inequalities that not only unifies many of the well-studied prophet inequality models under a single framework but also suggests natural ways to extend them. We obtain an algorithm that works for any problem in our framework, and its competitive ratio is a function of the graph parameters. The algorithm achieves a competitive ratio that is similar or equal to the state of the art for many of the problems we capture and even offers a better competitive ratio to what was formerly known for one of them.


\xhdr{Graph-Theoretic Framework}
 In this framework, we are given a connected directed acyclic graph (DAG) $G=(V,E)$, with starting node $s$ and target node $t$. Each edge $e\in E$ has a non-negative random (and possibly correlated) value $w_e$ sampled from a distribution $\mathcal{F}_e$. We require that the values of edges originating from different nodes be drawn independently, but we allow correlations between the values of edges originating from the same node.\footnote{{Correlated distributions in prophet inequality settings are known to result in notoriously bad competitive ratios \cite{immorlica2020prophet}, so it is quite surprising that in our framework we are able to handle some correlations without any deterioration in the competitive ratio.} } A decision-maker (DM) has to select a path from $s$ to $t$ sequentially, starting from $s$. The DM knows the distribution of the weights but only observes the realized weight of an edge $e=(u,v)$ if and when the node $u$ is reached. The objective of the DM is to maximize the expectation of the sum of weights along the selected path. An algorithm or policy $ALG$ is described by a rule {that dictates} at each node which edge to follow next, given the observed realizations so far. Let $ALG$ denote the sum of weights along the path that the agent selected.

We are interested in comparing $\E(ALG)$ with the expectation of the {offline} optimal path {(e.g., the one taken by an all knowing prophet)}. Formally, denoting by $\mathcal{P}$ the set of $s,t$-paths, let
\begin{align*}
    OPT= \max_{P\in \mathcal{P}} \sum_{e\in P} w_e.
\end{align*}
{We seek to understand} how large $\E(ALG)$ can be compared with $\E(OPT)$. 

\xhdr{Results: Graphs with One Focal Path} 
A basic instance in our framework is a graph that is composed of an $s,t$-path and some bypass edges that connect some of the nodes on the path {(as in Figures \ref{fig:classic-prophet} and \ref{fig:over-time})}. We refer to such a graph as a graph with one focal path. Our first main result is an algorithm for the single focal path problem that is $1/2$-competitive (\Cref{sec:one-focal}). The algorithm is based on a powerful technique known as online contention resolution~\cite{FSZ21}. The main idea is that, at each step, with constant probability, we stay in the focal path, and with constant probability we try to mimic the decision of the prophet. The reason for staying on the focal path is that from this path, we can observe all realizations. By staying on the path, we protect ourselves from mistakes arising from the fact that we can mimic the prophet only locally. A key observation is that the sum of the probabilities that the prophet selects the edges of any $s,t$-cut is $1$, so by scaling down the probabilities, we are guaranteed to stay in the focal path with constant probability.

Considering one focal path already suffices to capture the classic prophet inequality. To this end, we set the weight of each edge on the path to $0$ and add a bypass edge from each node $i=1,\dots, n$ to $t$. We set the weight of the bypass edge from $i$ to $t$ to be a sample from a known distribution $\mathcal{F}_i$. This instance is illustrated in Figure \ref{fig:classic-prophet}. The observation that our framework extends the classic prophet inequality problem immediately implies that our $1/2$-competitive algorithm is the best possible for this setting, as $1/2$ is already best-possible for the classic prophet inequality~\cite{krengel1978semiamarts,samuel1984comparison}. {Put differently, graphs with one focal path cover a broader range of applications than the traditional prophet inequality without sacrificing the quality of the competitive ratio.}

One focal path also suffices to capture the setting of Abels et al. \cite{abels2023prophet}, who study an over-time variant of the prophet inequality. In this variant, there is a set time horizon of $n$ steps and on each time step a candidate arrives and his value $v_i$ is sampled from some distribution $\mathcal F$. The values for all candidates are sampled IID from the same distribution. The DM at step $i$ chooses the number of time steps $t_i\geq 0$ to hire the candidate. The value for hiring the candidate for $t_i$ steps is $t_i\cdot v_i$. If the candidate is hired for $t_i$ steps, then the DM can only hire another candidate after $t_i$ steps, and therefore, cannot hire candidates between times {$i+1$} and {$i+t_i-1$}. In other words, if {$t_i>1$}, then the DM is forced to set $t_j=0$ for all $j$ in {$[i+1,i+t_i-1]$}. For the case that $n$ is fixed, Abels et al. present a $\approx 0.396$-competitive algorithm. This problem can be modeled using our framework with one focal path. To do so, we create a node on the path for each time step and add bypass edges corresponding to all possible hiring terms. We illustrate this in Figure \ref{fig:over-time}. Our algorithm for this setting is $1/2$-competitive, and this holds not only for the case where the values of all candidates are drawn from the same distribution but also for situations in which each value \(v_i\) is sampled from a different distribution \(\mathcal{F}_i\). Furthermore, the value for each employment term can be any function of \(v_i\) or sampled from an entirely different distribution that may or may not be correlated with \(\mathcal{F}_i\). This implies, for example, that we can handle cases where not all hiring terms are permitted; for instance, we may only be able to hire for 1, 3, or 12 months.




\begin{figure}[h!]
\centering
\begin{tikzpicture}[->,>=stealth,shorten >=1pt,auto,node distance=2.5cm,
                    thick,main node/.style={circle,draw,font=\sffamily\Large\bfseries}]

  \node[main node] (1) {$1$};
  \node[main node] (2) [right of=1] {$2$};
  \node[main node] (3) [right of=2] {$3$};
  \node[main node] (4) [right of=3] {$4$};
  \node[main node] (5) [right of=4] {$5$};
  \node[main node] (t) [right of=5] {$t$};

  \path[every node/.style={font=\sffamily\small}]
    (1) edge node[above] {0} (2)
    (2) edge node[above] {0} (3)
    (3) edge node[above] {0} (4)
    (4) edge node[above] {0} (5)
    (5) edge node[above] {$\mathcal{F}_5$} (t)
    (1) edge[bend left=40] node[above, pos=0.2] {$\mathcal{F}_1$} (t)
    (2) edge[bend left=40] node[above, pos=0.2] {$\mathcal{F}_2$} (t)
    (3) edge[bend left=40] node[above, pos=0.2] {$\mathcal{F}_3$} (t)
    (4) edge[bend left=40] node[above, pos=0.2] {$\mathcal{F}_4$} (t);
\end{tikzpicture}
\caption{Graph for capturing classic prophet inequality. The starting node is $s=1$.}
\label{fig:classic-prophet}
\end{figure}
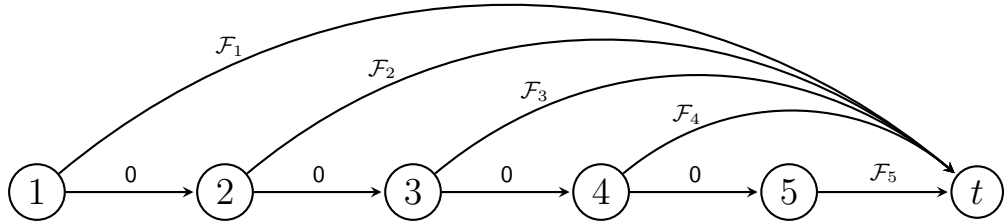

\begin{figure}[h!]
\centering
\begin{tikzpicture}[->,>=stealth,shorten >=1pt,auto,node distance=2.5cm,
                    thick,main node/.style={circle,draw,font=\sffamily\Large\bfseries}]

  \node[main node] (1) {$1$};
  \node[main node] (2) [right of=1] {$2$};
  \node[main node] (3) [right of=2] {$3$};
  \node[main node] (4) [right of=3] {$4$};
  \node[main node] (5) [right of=4] {$5$};
  \node[main node] (t) [right of=5] {$t$};

  \foreach \from/\to in {2/4,2/5,2/t,3/5,3/t,4/t}
    \path[draw=gray!50] (\from) edge[bend left=40] (\to);

  \foreach \from/\to in {1/2,2/3,3/4,4/5,5/t}
    \path (\from) edge (\to);

  \foreach \from/\to in {1/3,1/4,1/5,1/t}
    \path (\from) edge[bend left=40] (\to);

\end{tikzpicture}
\caption{An illustration for the over-time variant of prophet inequality with a time horizon of $5$ steps. Some of the edges are gray to make the illustration easier to digest. }
\label{fig:over-time}
\end{figure}
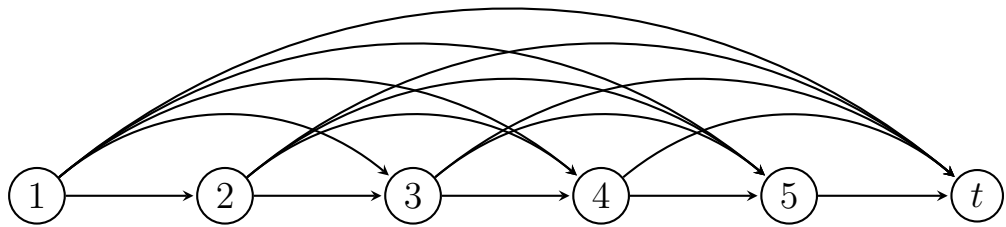

\xhdr{Results: General Graphs}
We identify a key property of the graph that determines the achievable competitive ratio. This is the number of paths that are required to cover all the nodes of the graph. By Dilworth's theorem~\cite{D50}, this number is the same as the size of the maximal antichain in the graph. We observe an {upper} bound of $\frac{1}{k+1}$ where $k$ is the number of paths required to cover the graph and design a $\frac{1}{2k}$-competitive algorithm (\Cref{sec:general}). Roughly speaking, this is done by carefully separating the graph into $k$ focal paths in a way that guarantees that at least one of the paths collects a $\frac{1}{k}$ fraction of the optimal solution. Next, we use our algorithm for a single focal path as a black box to achieve $1/2$ of the expected value of the optimal solution on each separate path. Putting this together, we get a $\frac{1}{2k}$-competitive algorithm. {For the special case of graphs that are composed of $k$ disjoint paths we obtain  an improved competitive ratio of $\frac{1}{k+1}$.}
 We note that it is possible to have different path covers of the same size. The exact choice of the cover can significantly impact the performance of the algorithm. We highlight an instance in which there exists one path cover for which our algorithm obtains $\frac{1}{2k}$ of the expected value of the optimal solution and another path cover on which it obtains a $\frac{1}{k+1}$ fraction.

Work on prophet inequalities usually focuses on problems that are downward-closed. That is, any subset of a feasible solution is also a feasible solution. Our graph-theoretic framework is {unique} in the landscape of prophet inequality problems in that even though the set of feasible solutions is not downward-closed we are able to design algorithms with good competitive ratio. 

The framework covers useful extensions of well studied problems to a setting in which the DM has $k$ markets at his disposal and can switch from one market to another. As a concrete example, consider again the setting of over-time prophet inequality by Abels et al. \cite{abels2023prophet}. Now, assume that the DM has $k$ different markets. For each market there is a different known sequence of candidates. When the DM chooses the length of a term to hire a candidate he can also choose which market to turn to at the end of this period. An illustration of this can be found in Figure \ref{fig:markets}.

\begin{figure}[h!] 
  \centering
  \begin{tikzpicture}[>=stealth,shorten >=1pt,initial/.style={}, every node/.style={circle,draw}]
    \node (s) at (0,0) {$s$};
    \node (v1) at (2,-1) {$v_1$};
    \node (v2) at (4,-1) {$v_2$};
    \node (v3) at (6,-1) {$v_3$};
    \node (v4) at (8,-1) {$v_4$};
    \node (u1) at (2,1) {$u_1$};
    \node (u2) at (4,1) {$u_2$};
    \node (u3) at (6,1) {$u_3$};
    \node (u4) at (8,1) {$u_4$};
    \node (t)  at (10,0) {$t$};

    \draw[->] (s) -- (v1);
    \draw[->] (v1) -- (v2);
    \draw[->] (v1) to[bend right] (v3); 
     \draw[->] (v2) to[bend right] (v4); 
      \draw[->] (v3) to[bend right=60] (t); 
    \draw[->] (v1) -- (u2);
    \draw[->] (v2) -- (v3);
    \draw[->] (v2) -- (u3);
    \draw[->] (v3) -- (v4);
    \draw[->] (v4) -- (t);

    \draw[->] (s) -- (u1);
    \draw[->] (u1) -- (u2);
    \draw[->] (u2) -- (u3);
    \draw[->] (u3) -- (u4);
    \draw[->] (u4) -- (t);
     \draw[->] (u1) to[bend left] (u3); 
     \draw[->] (u2) to[bend left] (u4); 
      \draw[->] (u3) to[bend left=60] (t); 

    \draw[->] (u1) -- (v2);
    \draw[->] (u2) -- (v3);
    \draw[->] (u3) -- (v4);
    \draw[->] (v1) -- (u3);
    \draw[->] (v2) -- (u4);
    \draw[->] (u1) -- (v3);
    \draw[->] (u2) -- (v4);
    \draw[->] (u3) -- (t);
    \draw[->] (v3) -- (t);
    \draw[->] (v3) -- (u4);
  \end{tikzpicture}
  \caption{An illustration for over-time prophet inequality with two markets. To simplify the illustration the candidates can be hired for either one time step or two. From $s$, the DM can decide which market to go to first so both edges have weight $0$. For any $i$ the weight of the edges $(v_i,v_{i+1}), (v_i,u_{i+1})$ is the same and the weight of the edges $(v_i,v_{i+2}), (v_i,u_{i+2})$ is the same.}
 \label{fig:markets}
\end{figure}
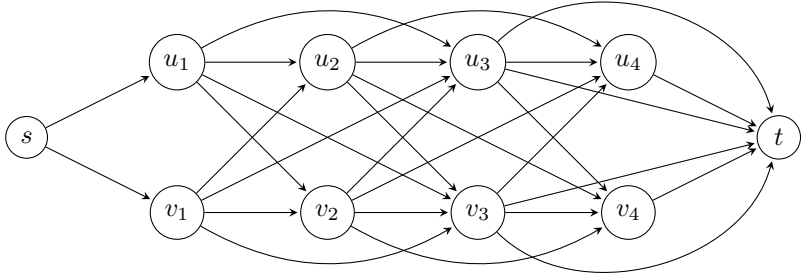

\xhdr{Extension: Edge Labels and Capacities} 
{Our basic graph theoretic framework does not cover }problems such as the multiple-choice prophet or online sequential matching without a considerable blow-up in the graph size. To remedy this situation, and to allow for a much more expressive framework, we augment the graph by adding edge labels and capacities on the number of edges that the DM can use for each label. We impose the same constraints on the prophet. In this general framework, each edge can have multiple labels, and we use $d$ to denote the maximum number of labels that an edge may have. Finally, we assume that any two nodes that are connected by a labeled edge are also connected by an unlabeled edge.


For the case of one focal path and at most $d$ labels on each edge, we obtain an algorithm that is $\frac{1}{d+2}$-competitive (\Cref{sec:one-focal-labels}). Observe that with a single label (e.g., $d=1$) we can already capture the multiple-choice prophet inequality~\cite{HKS07,A14}, {in which the DM selects at most $m$ candidates.} This can be done by using a single path and having two parallel edges from each node to the subsequent one. One edge without any label has a fixed value of $0$ and the other is a labeled edge with the value sampled from some known distribution. We also set the capacity for this label to be $m$. For one label ($d=1$), we have a $\frac{1}{3}$-competitive algorithm. This is arguably far from the optimal factor of $1-O(1/\sqrt{m})$ for the $m$-choice prophet inequality~\cite{A14}, but recall that already for $d=0$ the model captures the classic prophet inequality so the competitive factor cannot be better than $1/2$.
In fact, for the case of $d=1$ we provide an upper bound of $\frac{4}{9}$ on the best achievable competitive ratio. 

Adding labels also allows us to handle the problem of online stochastic bipartite max-weight matching \cite{gravin2019prophet}. In this problem, we have a set of items and bidders. Each bidder is interested in a single item and its value for each item is sampled from a known distribution. In the edge arrival variant of the problem, at each time step a value of a bidder for an item is revealed and the DM needs to decide whether to allocate the item to the bidder or not. To capture this setting in our framework, we define the set of labels to include all bidders and items. Each label has a capacity of $1$ and each edge is labeled by the bidder and item it connects, so $d=2$. For this problem our algorithm is $1/4$-competitive, which is not too far apart from the best known algorithm which is $0.344$-competitive \cite{MMG23,EFGT23}. We can also model the one sided vertex arrival variant of the problem. In which in each time step a bidder arrives and its values for all the items are realized. {For this variant, a single label on each edges suffices: we label each edge with its end point and set a capacity of $1$ for each label. Hence, we obtain a $\frac{1}{3}$-competitive algorithm}. 
This is in comparison to the optimal competitive ratio which is $1/2$ \cite{feldman2014combinatorial}. Lastly, by allowing for each edge to have $d$ labels we can capture the yet more general problem of $d$-dimensional matching. For this problem the best known algorithm is $1/(d+1)$-competitive~\cite{CCFPW23}. {We would like to highlight that our bounds are slightly worse than the best known bounds for these specific problems for a good reason: our framework is strictly more general. For example, it allows to seamlessly add an outside option by connecting $s$ directly with $t$.} 


\subsection{Related Work}

Beyond the classic prophet inequality setting of Krengel and Sucheston~\cite{krengel1978semiamarts}, a large body of literature considers combinatorial extensions of the model, motivated by its connections with mechanism design. A basic extension is that of the multi-choice prophet inequality, where the DM selects at most $m$ random variables instead of just one~\cite{A14,CDL23,JMZ22}, where the best possible competitive ratio is $1-\frac{1}{\sqrt{m+3}}$. Another important extension is the case of selecting an independent set in an underlying matroid. Interestingly, as in the classic setting, there is a $1/2$-competitive algorithm for this case~\cite{KW12}. If the DM selects a matching in a graph, there is a $0.344$-competitive algorithm~\cite{MMG23} if edges arrive one by one, and a $1/2$-competitive algorithm if nodes with all their incident edges arrive one by one~\cite{feldman2014combinatorial}. 

Relevant to the connection with mechanism design is the case of online combinatorial auctions, where a sequence of buyers arrive one by one and the DM has a set of $m$ heterogeneous items to allocate. Each buyer has a random valuation function independent of the others. A DM has to design a mechanism to allocate the items to maximize social welfare, which is equivalent to designing an online algorithm that decides which set of items each buyer gets. For the case that valuations belong to a class known as XOS, there is a $1/2$-competitive algorithm~\cite{feldman2014combinatorial}. A recent result shows that if valuations are subadditive (which contains the class XOS), then there is a $1/6$-competitive algorithm~\cite{CC23}. If valuations are not subadditive, no better than an $O(\log m/ \sqrt{m})$ factor is possible~\cite{MSV23}. However, if no buyer is allowed to take more than $d$ items, there is a $\frac{1}{d+1}$-competitive algorithm.

The line of work on time-inconsistent planning originating from \cite{kleinberg2014time} operates is a similar graph theoretic framework.
These works consider a similar model on a directed acyclic graph, except that the rewards on the edges are deterministic. The sequential aspect comes from the assumption that the DM exhibits a type of bias known as \emph{present bias}. This means the DM perceives the reward collected at the current time step as having a higher value (e.g., multiplied by some parameter $b>1$) than its nominal value. As a result, at each node, the DM recomputes the path she plans to take. This correspondence raises a fascinating direction of understanding prophet inequalities for present-biased decision-makers.

\section{Model}
We are given a directed acyclic graph (DAG) $G=(V,E)$, with starting node $s$ and target node $t$. Each edge $e\in E$ has a non-negative random (and possibly correlated) value $w_e$. A decision-maker (DM) has to select a path from $s$ to $t$ sequentially, starting from $s$. The DM knows the distribution of the values but only observes the realized weight of an edge $e=(u,v)$ if and when the node $u$ is reached. The objective of the DM is to maximize the expectation of the sum of values along the selected path. An algorithm or policy $ALG$ determines at each node which edge to follow next, given the observed realizations so far. We slightly abuse the notation and use $ALG$ to refer to both the set of selected edges and the sum of their values.

We compare $\E(ALG)$ with the expectation of the optimal path. {Formally, if we denote  $\mathcal{P}$ the set of $s,t$-paths, the optimal path is}
\begin{align*}
    OPT= \max_{P\in \mathcal{P}} \sum_{e\in P} w_e.
\end{align*}
We are interested in understanding how large is $\E(ALG)$ compared with $\E(OPT)$. We refer to $OPT$ as the value of the optimal offline algorithm, or the value of a prophet, as it is the value the DM would get if all values were revealed beforehand.
We say $ALG$ is $c$-competitive if $\E(ALG)\geq c\cdot \E(OPT)$.

\paragraph{Correlation.} We assume that edges that start in different nodes are independent. Edges that start in the same node can be arbitrarily correlated.

\paragraph{Path Covering.} It will be important for our results to define a \emph{path covering} of the graph. For a path $P\in \mathcal{P}$, we denote by $V(P)$ the set of nodes visited by $P$. We say a set of paths $P_1,\dots, P_k\in \mathcal{P}$ covers the graph if $V=\cup_{i=1}^k V(P_i)$. We say the graph has \emph{width} $k$ if the smallest number of $s,t$-paths needed to cover it is $k$. A related notion is that of an \emph{antichain}. A set $A\subseteq V$ is an antichain if, for every pair of nodes $u,v\in A$; there is no path connecting $u$ to $v$ or $v$ to $u$. The following observation connects these two notions.
\begin{theorem}[Dilworth's Theorem~\cite{D50}]
    A {directed acyclic} graph has width $k$ if and only if its largest antichain has size $k$.
\end{theorem}

\paragraph{Edge Labels.} We extend our model by introducing labels for the edges. We are given a finite set of labels $L$, and for each edge $e\in E$, there is a set of labels $L_e\subset L$. For a label $\ell\in L$ we denote as $E(\ell) = \{e\in E: \ell \in L_e\}$ the set of edges with label $\ell$. For each label, we are given a capacity $c_\ell \geq 1$, and we add the constraint that the DM cannot take more than $c_\ell$ edges from $E(\ell)$, i.e., for an $s,t$-path $P$ to be feasible, we require that $|P\cap E(\ell)|\leq c_{\ell}$ {for every $\ell\in L$ }. We also impose that $OPT$ satisfies the same constraint. We assume that for every labeled edge, there is a parallel unlabeled edge.\footnote{This is a technical assumption that prevents an online algorithm that makes local decisions from getting trapped because of the labels. We can allow for a slightly more general assumption. For every pair of nodes $u,v$, if $v$ can be reached from $u$, then there is a way to reach $v$ from $u$ that uses only unlabeled edges.}
{Unless we} explicitly mention that we are considering the model with labels, we consider the model without labels.

\section{Width 1 Graphs: Single Focal Path with Bypasses}
In this section, we consider graphs that are composed of a single path with optional bypasses as in Figure \ref{fig:classic-prophet}. We first present a $1/2$-competitive algorithm for graphs with unlabeled edges and then adapt the algorithm to achieve a competitive ratio of $\frac{1}{d+2}$ for the more general setting of edges with at most $d$ labels. While the latter algorithm also covers unlabeled graphs, we find it instructive to present our algorithm for graphs with unlabeled edges first, as it is easier to comprehend. 

\subsection{Unlabeled Edges} \label{sec:one-focal}

In this section, we prove the following result.

\begin{theorem}
\label{thm:width1_nolabel}
    If the graph has width $1$, i.e., there is an $s,t$-path such that $V(P)=V$, then there is an algorithm $ALG$ that satisfies $\E(ALG)\geq \frac{1}{2} \E(OPT)$.
\end{theorem}

Before proving the theorem, we show that the case of graphs of width 1 already captures the classic prophet inequality model as a special case. In the classic model, we have values $X_1,\dots, X_n$, and we can select at most one. This corresponds to a graph with $n+1$ nodes, $V=\{1,2,\dots,n+1\}$, where $s=1$ and $t=n+1$. Each node $i\leq n$ is connected to node $i+1$ via an edge with value $w_{i,i+1}=0$, and to $t$ via an edge with value $w_{i,t}=X_i$. It is easy to see that this instance is equivalent to the classic prophet inequality {(Figure \ref{fig:classic-prophet})}. It is well known that no algorithm achieves a better competitive factor than $1/2$, so this reduction shows that \Cref{thm:width1_nolabel} is tight.

\paragraph{Algorithm.} For $e$ in $E$, let $x_e$ be the probability that $OPT$ takes edge $e$. Enumerate the nodes in the order $P$ visits them as $1,2,\dots,n$. At a node $i$, the algorithm observes the realizations of the edges that start at $i$ and tentatively chooses one of them with the same probability as $OPT$, conditional on the values of the edges that start at $i$ (but independently of previously observed values). We refer to this edge as the \emph{tentative edge}. Let $E_{-i}$ be the set of edges that skip $i$, i.e., edges $(j,k)$ such that $j<i<k$. The algorithm takes the tentative edge independently, with probability 
\[\alpha(i)= \frac{1/2}{1-\sum_{e\in E_{-i}} \frac{x_e}{2}},\] and with probability $1-\alpha(i)$, it stays in the path, i.e., it takes the edge $(i,i+1)$. If there is no tentative edge,\footnote{When we draw a realization of $OPT$, conditional on the values of edges that start at $i$, it can happen that the realized path does not visit node $i$, so there is no tentative edge in this step.} or if the tentative edge is $(i,i+1)$, then the algorithm just takes $(i,i+1)$. Notice that $OPT$ cannot use more than one edge in $E_{-i}$ in any realization, so $\alpha(i)\in [0,1]$, and therefore, $ALG$ is well-defined.

\begin{lemma}
    For every $e\in E\setminus P$, the probability that $e$ is marked as tentative and is selected by $ALG$ is exactly $x_e/2$.
    For every $e\in P$ this probability is at least $x_e/2$.
\end{lemma}
\begin{proof}
We show inductively that $ALG$ takes each edge outside the path with probability exactly $x_e/2$ (notice that outside $P$, $ALG$ only takes tentative edges). By definition, $E_{-1}=\emptyset$, so $\alpha(1)=1/2$. Therefore, the statement is true for all edges that start from node $1$. Consider a node $i$ and assume that the statement is true for all edges that start from nodes $1,\dots,i-1$. By the inductive hypothesis, the probability that $ALG$ visits $i$ is exactly
$$1-\sum_{e\in E_{-i}} \frac{x_e}{2},$$
since $ALG$ either uses exactly one edge from $E_{-i}$ and does not visit node $i$, or it does visit node $i$ and does not use any edge in $E_{-i}$.

Notice that the event that $ALG$ visits $i$ is independent of the values of edges that start at $i$.  Therefore, for an edge $e$ that starts at $i$ and does not belong to $P$, the probability that $ALG$ takes $e$ is exactly the probability of visiting $i$, times $x_e$, times $\alpha(i)$, i.e.,
$$\left( 1-\sum_{e'\in E_{-i}} \frac{x_{e'}}{2}\right)\cdot x_e \cdot \alpha(i) = x_e/2,$$
by the definition of $\alpha(i)$. If $e\in P$, then the probability that $e$ is tentative and $ALG$ takes it is
$$\left( 1-\sum_{e'\in E_{-i}} \frac{x_{e'}}{2}\right)\cdot x_e \geq \left( 1-\sum_{e'\in E_{-i}} \frac{x_{e'}}{2}\right)\cdot x_e \cdot \alpha(i) = x_e/2.$$
\end{proof}

\begin{proof}[Proof of \Cref{thm:width1_nolabel}.]
    We have that
\begin{align*}
\E(ALG)&=\sum_{e\in E} \E\left( w_e\cdot \mathds{1}_{\{ ALG \text{ takes } e \} } \right )\\
&\geq 
\sum_{e\in E} \E\left( w_e\cdot \mathds{1}_{\{ ALG \text{ marks as tentative and takes } e \} } \right)\\
&=\sum_{e\in E} \E(w_e\,|\,ALG \text{ marks as tentative and takes } e)\\
&\qquad \qquad \qquad \qquad \cdot \PP(ALG \text{ marks as tentative and takes } e)\\
& = \sum_{e\in E} \E(w_e\,|\,ALG \text{ marks } e  \text{ as tentative})\cdot \frac{x_e}{2}\\
&= \frac{1}{2}\sum_{e\in E} \E(w_e \,|\, OPT \text{ takes } e) \cdot x_e\\
&= \frac{1}{2}\E(OPT).
\end{align*}
We do not have equality because $ALG$ takes edges in $P$ with a larger probability than $x_e/2$.

\end{proof}

\subsection{Width 1 Graphs with Edge Labels}  \label{sec:one-focal-labels}

In this section, we consider the extension of our model where each edge also has a set of labels. 

\begin{theorem}
    \label{thm:colors}
    If each edge has at most $d$ labels, then there is an algorithm $ALG$ such that $\E(ALG)\geq \frac{1}{d+2}\E(OPT)$.
\end{theorem}

While we leave the question open the question of the tightness of this theorem, 
for $d=1$, we present an upper bound of $4/9$ on the competitive ratio at the end of this section (Example \ref{ex:49k1d1}). This separates our setting from the case where we only care about the capacities and we don't require the solution to be a path, for which there is a $1/2$-competitive algorithm for $d=1$.

For each edge $e$ in $E$, let $x_e$ be the probability that $OPT$ takes edge $e$. Let $P$ denote an unlabeled path that visits all nodes. We enumerate nodes in the order they are visited by $P$.

\paragraph{Algorithm.} We describe an algorithm $ALG$ that takes every edge $e$ outside $P$ with probability exactly $x_e/(d+2)$, and every edge $e$ in $P$ with probability at least $x_e/(d+2)$. To do so, at any node $i$, the algorithm observes the values of all edges that leave $i$. Then, it samples $OPT$ conditional on the given values of the edges that leave $i$. Call this $OPT(i)$. If $OPT(i)$ uses one of the edges that leave $i$, $ALG$ tentatively selects that edge. Otherwise, $ALG$ takes the edge in $P$ that leaves $i$. Let $e$ be the tentatively selected edge, and denote by $p(e)$ the probability that it is feasible to take $e$, i.e., the probability that $ALG$ visits node $i$, and at that point it has available capacity for each of the labels in $L_e$. Given that $e\in OPT(i)$ and that it is feasible to take it, $ALG$ takes $e$ with probability $\frac{1}{d+2}\cdot \frac{1}{p(e)}$, independently of everything else. Otherwise, $ALG$ takes the next edge in $P$. 
\begin{lemma}
    \label{lem:lemcolors}
    $ALG$ is well-defined and for every $e\in E\setminus P$, the probability that $e$ is marked as tentative and is selected by $ALG$ is exactly $x_e/(d+2)$.
    For every $e\in P$ this probability is at least $x_e/(d+2)$.
\end{lemma}
\begin{proof}

    We need to show that $\frac{1}{d+2}\cdot \frac{1}{p(e)}\in [0,1]$ and that $\PP(e\in ALG) = x_e/(d+2)$ for every edge outside $P$. We prove this by induction over the nodes, in the order they are visited by $P$. For the first node, this is immediate because $p(e)=1$ for every edge $e$ that starts at the first node, and $e$ is tentative with probability $x_e$. For a node $i$, assume the statement is true for every edge leaving nodes $i'<i$. Consider an edge $e$ that leaves $i$. Let $A_i$ be the event that $ALG$ skips node $i$, i.e., that it takes an edge that starts in a node $i'<i$ and ends in a node $i''>i$. Let $B_{i,\ell}$ be the event that at node $i$ (or when skipping node $i$) $ALG$ saturated the capacity constraint for label $\ell$. We know that if none of these events occur, then it is feasible to select $e$. Thus, by a union bound,
    \begin{align}
        p(e) \geq 1-\PP(A_i)-\sum_{\ell\in L_e} \PP(B_{i,\ell}).
        \label{eq:unionbound_labeled}
    \end{align}
    Notice that the edges that skip $i$ belong to an $s,t$-cut in the graph, so both $OPT$ and $ALG$ can use at most one of them. Also, all of them start in nodes $i'<i$, so by the induction hypothesis, $ALG$ visits any of those edges $e'$ with probability $x_{e'}/(d+2)$. These two facts imply that 
    $$\PP(A_i)=\sum_{e' \text{ skips } i} \frac{x_{e'}}{d+2}\leq \frac{1}{d+2}.$$ 
    Let  $X_{i,j}$ denote the number of edges of label $j$ selected by $ALG$ when visiting or skipping $i$. By the induction hypothesis, $\E(X_{i,j})\leq c_j/(d+2)$. Therefore, by Markov's inequality, $\PP(B_{i,j})=\PP(X_{i,j}\geq c_j)\leq 1/(d+2)$. Since $|L_e|\leq d$, from \cref{eq:unionbound_labeled} we conclude that $p(e)\geq 1/(d+2)$. This implies that $\frac{1}{d+2}\cdot \frac{1}{p(e)} \in [0,1]$. Finally, since whether $e$ is feasible and whether it is tentatively selected are independent events, we have that
    \begin{align*}
        \PP(e\in ALG)= x_e \cdot p(e) \cdot \frac{1}{d+2}\cdot \frac{1}{p(e)} = \frac{x_e}{d+2}.
    \end{align*}
\end{proof}

\begin{proof}[Proof of \Cref{thm:colors}.]
%

We have that
\begin{align*}
\E(ALG)&=\sum_{e\in E} \E\left( w_e\cdot \mathds{1}_{\{ ALG \text{ takes } e \} } \right )\\
&\geq 
\sum_{e\in E} \E\left( w_e\cdot \mathds{1}_{\{ ALG \text{ marks as tentative and takes } e \} } \right)\\
&=\sum_{e\in E} \E(w_e\,|\,ALG \text{ marks as tentative and takes } e)\\
&\qquad \qquad \qquad \qquad \cdot \PP(ALG \text{ marks as tentative and takes } e).
\end{align*}
Using \Cref{lem:lemcolors}, and the fact that given that $e$ is marked as tentative, $w_e$ is independent of the event that $ALG$ selects $e$, we have that
\begin{align*}
\E(ALG)
& \geq \sum_{e\in E} \E(w_e\,|\,ALG \text{ marks } e  \text{ as tentative})\cdot \frac{x_e}{d+2}\\
&= \frac{1}{d+2}\sum_{e\in E} \E(w_e \,|\, OPT \text{ takes } e) \cdot x_e\\
&= \frac{1}{d+2}\E(OPT).
\end{align*}
\end{proof}
\subsection{An Upper Bound for Width 1 Graphs with One Label}
To complement our algorithmic result, we observe the following upper bound for $d=1$:
\begin{example}[Upper bound of $4/9$ for $k=1$ and $d=1$] \label{ex:49k1d1}
    Consider a graph with a single focal path, where apart from $s$ and $t$, we have two internal nodes, $a$ and $b$, as in \Cref{fig:49k1d1}. We have a deterministic edge with value $2$ from $s$ to $t$. We also have an edge from $a$ to $t$ with value $2$ with probability $1/2$ and $0$ otherwise. We have a single label, \emph{red}, with a capacity of $1$, and two edges with this label. The first one is an edge from $s$ to $a$ with value $1$, and the second is an edge from $b$ to $t$ with value of $2/\varepsilon$ with probability $\varepsilon$, and $0$ otherwise.
    
    When the value $2/\varepsilon$ is realized,  $OPT=2/\varepsilon$. This happens with probability $\varepsilon$. When it is not, and the other random value is realized at $2$, then $OPT$ takes the first red edge, and the realized value of $2$, so $OPT=1+2$. This happens with probability $(1-\varepsilon)/2$. The rest of the time, $OPT=2$. Therefore, we have that $\E(OPT)=(2/\varepsilon)\cdot \varepsilon +3\cdot (1-\varepsilon)/2 + 2\cdot (1-\varepsilon)/2  = 9/2 - O(\varepsilon)$. For an online algorithm, notice that we are indifferent between the three feasible options, all having a value of $2$. Thus, we obtain an upper bound of $4/9$.
\end{example}

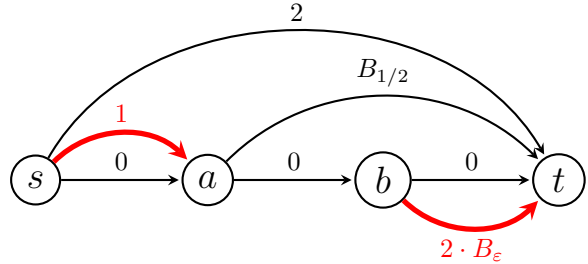
\begin{figure} [htb]
    \centering

\begin{tikzpicture}[->,>=stealth,shorten >=1pt,auto,node distance=1.6cm,
                    thick,node_style/.style={circle,draw,font=\sffamily\Large\bfseries}]

    \node[node_style] (s) {$s$};
    \node[node_style, right = of s] (i1) {$a$};
    \node[node_style, right = of i1] (i2) {$b$};
    \node[node_style, right = of i2] (t) {$t$};

    \draw[->] (s) -- (i1) node[midway, above] {0};
    \draw[->] (i1) -- (i2) node[midway, above] {0};
    \draw[->] (i2) -- (t) node[midway, above] {0};

    \draw[->, red, line width=2pt] (s) to[bend left=45] node[above] {$1$} (i1);
    \draw[->] (i1) to[bend left=45] node[above] {$B_{1/2}$} (t);
    \draw[->, red, line width=2pt] (i2) to[bend right=45] node[below] {$2\cdot B_\varepsilon$} (t);

    \draw[->] (s) to[bend left=60] node[above] {$2$} (t);
\end{tikzpicture}
    
    \caption{Instance from \Cref{ex:49k1d1} that attains the upper bound of $4/9$ for graphs with a single focal path and $d=1$. $B_x$ denotes an independent realization of a variable that has value $1/x$ with probability $x$, and $0$ otherwise. There is only one label, red, and the thick red arrows represent edges with this label. The capacity of the red label is $1$.}
    \label{fig:49k1d1}
\end{figure}

\section{General Graphs} \label{sec:general}

In this section, we consider the case of general graphs, with the following theorem as our main result.
\begin{theorem}
\label{thm:general_graphs}
    If the graph has width $k$ and each edge has at most $d$ labels, then there is an algorithm $ALG$ that satisfies $\E(ALG)\geq \frac{1}{k\cdot (d+2)}\E(OPT)$.
\end{theorem}

We now briefly discuss the computation of a path cover, and then we prove the main theorem. Next, we show that if the graph can be covered with $k$ node-disjoint paths, the competitive ratio of our algorithm is $\frac{1}{k+1}$. Finally,
we present a lower bound of $1/(k+1)$ and show that depending on the specific covering we choose, our analysis is tight for $d=0$.

\paragraph{Computation of the Path Covering} A classic result of Dilworth~\cite{D50} and Fulkerson~\cite{F56} states that we can compute a minimum path cover of a DAG in polynomial time. A simple way to do so is to find the transitive closure of the graph, then construct a bipartite graph where for each node $v\in V$, on each side, we have copies $v_\ell$ and $v_r$. For each edge $(u,v)$ in the original graph, we add an edge $(u_\ell,v_r)$. A minimum path cover corresponds to a maximum matching in this new graph. Using the Hopcroft-Karp algorithm~\cite{HK73} this gives a running time of $O(|V|^{2.5})$. If we allow for dependency on $k$, a recent result states that it can be done in time $O(k^2|V|+|E|)$~\cite{CCMRT22}.

\subsection{Proof of \Cref{thm:general_graphs} and Tightness}
\begin{proof}
    If the graph has width $k$, then there are paths $P_1,\dots, P_k$ that cover all nodes. We construct graphs $G_1,\dots, G_k$ as follows. In graph $G_i$, we take the graph that has $V(P_i)$ as the set of nodes and has all edges between nodes in $V(P_i)$. Then, for every edge $e$ that starts at a node $u\in V(P_i)$ and ends in a node $v\not\in V(P_i)$, let $v'$ be the earliest node in $V(P_i)$ (in the ordering induced by $P_i$) such that there is a path from $v$ to $v'$. Note that there is always a path from $v$ to $t$, and $t\in V(P_i)$, so there is always a node $v'$. We add an edge $e'=(u,v')$ with value $w_{e'}:= w_e$. We call these edges \emph{artificial}. Note that with this construction, every edge that starts in a node in $V(P_i)$ is represented. See \Cref{fig:pathwdetours_thm5} for an illustration of the construction. Let $E^+(P_i)$ be the set of edges that start in a node in $V(P_i)$. There is always a path in $G_i$ with a value at least as large as the value of the edges in $OPT\cap E^+(P_i)$. Let $OPT(G_i)$ denote the optimal path in the graph $G_i$. Since all edges in the original graph are represented by some edge in one of the new graphs, we have that
    \begin{align*}
        \E(OPT)\leq \sum_{i=1}^k \E(OPT(G_i)).
    \end{align*}

    Consider the algorithm that chooses uniformly at random one of the instances $G_1,\dots, G_k$. Let this be $G_i$. Then, it runs the algorithm given by \Cref{thm:colors} on $G_i$. That is, whenever the algorithm takes an edge that starts and ends in $V(P_i)$, it takes the same edge in the original graph; and whenever it takes an artificial edge, it takes the corresponding path in $G$. By the definition of $G_i$, and because choosing $G_i$ is independent of the values, we have that conditional on choosing $G_i$, $ALG$ gets as much value as the algorithm for a single path on $G_i$, which by \Cref{thm:colors} is at least $\frac{1}{d+2}\cdot \E(OPT(G_i))$. Therefore, we have that 
    \begin{align*}
        \E(ALG) &\geq \frac{1}{k}\cdot \sum_{i=1}^k \frac{1}{d+2} \cdot \E(OPT(G_i))\\
        &\geq \frac{1}{k\cdot (d+2)} \cdot \E(OPT).
    \end{align*}
    This concludes the proof of the theorem.
\end{proof}

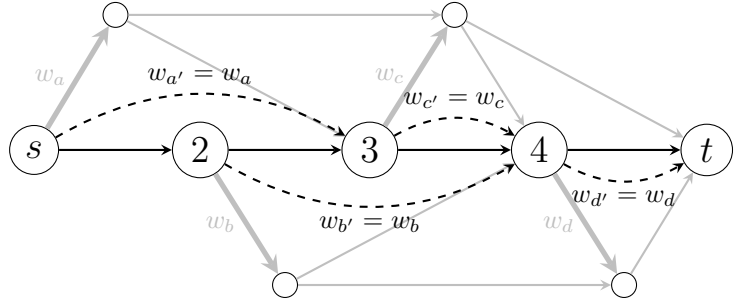
\begin{figure}
    \centering


\begin{tikzpicture}[node distance=1.5cm and 1.5cm, 
                    node_style/.style={circle,draw,font=\sffamily\Large\bfseries},
                    main_path/.style={->,>=stealth, thick},
                    detour_path/.style={->,>=stealth,gray!50, thick},
                    detour_path_thicker/.style={->,>=stealth,gray!50, line width=2pt},
                    new_path/.style={->,>=stealth,dashed, thick}]

    \node[node_style] (i1) {$s$};
    \node[node_style, right = of i1] (i2) {$2$};
    \node[node_style, right = of i2] (i3) {$3$};
    \node[node_style, right = of i3] (i4) {$4$};
    \node[node_style, right = of i4] (t) {$t$};

    \path (i1) -- (i2) node[midway] (posd1) {};
    \path (i3) -- (i4) node[midway] (posd2) {};
    \path (i2) -- (i3) node[midway] (posd3) {};
    \path (i4) -- (t) node[midway] (posd4) {};
    
    \node[node_style, above = of posd1] (d1) {};
    \node[node_style, above = of posd2] (d2) {};
    \node[node_style, below = of posd3] (d3) {};
    \node[node_style, below = of posd4] (d4) {};

    \draw[main_path] (i1) -- (i2);
    \draw[main_path] (i2) -- (i3);
    \draw[main_path] (i3) -- (i4);
    \draw[main_path] (i4) -- (t);
    
    \draw[detour_path_thicker] (i1) -- (d1) node[midway, left] {$w_a$};
    \draw[detour_path] (d1) -- (i3);
    \draw[detour_path_thicker] (i2) -- (d3) node[midway, left] {$w_b$};
    \draw[detour_path] (d3) -- (i4);
    \draw[detour_path] (d1) -- (d2);
    \draw[detour_path] (d2) -- (t);
    \draw[detour_path] (d2) -- (i4);
    \draw[detour_path] (d3) -- (d4);
    \draw[detour_path] (d4) -- (t);
    \draw[detour_path_thicker] (i3) -- (d2) node[midway, left] {$w_c$};
    \draw[detour_path_thicker] (i4) -- (d4) node[midway, left] {$w_d$};

    \draw[new_path] (i1) to[bend left]  node[midway, above] {$w_{a'}=w_a$}(i3);
    \draw[new_path] (i2) to[bend right] node[midway, below] {$w_{b'}=w_b$}(i4);
    \draw[new_path] (i3) to[bend left] node[midway, above] {$w_{c'}=w_c$}(i4);
    \draw[new_path] (i4) to[bend right] node[midway, below] {$w_{d'}=w_d$}(t);


\end{tikzpicture}

    \caption{An illustration of the construction of a graph $G_i$ in the proof of \Cref{thm:general_graphs}. We start with a path $P_i=(s,2,3,4,t)$ and replace edges that leave the path (thick grey arrows, $a,b,c$, and $d$) with an artificial edge (dashed arrows, $a',b',c'$, and $d'$) that ends in the earliest node in $P_i$ that can be reached after taking the original edge. The values of the artificial edges are the same as in the original edges, and we ignore the values of the intermediate edges that do not start from nodes in $P_i$.}
    \label{fig:pathwdetours_thm5}
\end{figure}

To complement \Cref{thm:general_graphs}, we show that depending on which path covering we choose, our analysis can be tight: 


\begin{example}[Tightness of the analysis in \Cref{thm:general_graphs}] \label{ex:grid}
    Consider a graph with the form of a $k\times k$ grid, where the top-left corner is the source $s$ and the bottom-right corner is the sink $t$ as in \Cref{fig:lowerbound_grid}. All horizontal edges are oriented to the right, and all vertical edges are oriented down. Each edge in the second vertical column has a deterministic value of $1$. Add an edge from the last node in the second vertical column to $t$ with a value of $1$. Next, for a given $\varepsilon>0$, each edge in the third vertical column has an independent random value of $1/\varepsilon$ with probability $\varepsilon$ and $0$ otherwise. Add an edge from the last node of the third column to $t$ with a value drawn from the same distribution. All other edges (vertical and horizontal) have a deterministic value of $0$. Notice that $OPT$ can take all the realized  $1/\varepsilon$'s. When none is realized, $OPT$ can take all the deterministic $1$'s. Therefore, $\E(OPT)\geq k\cdot \varepsilon/\varepsilon+ k\cdot (1-k\varepsilon)= 2k-O_k(\varepsilon)$. If we cover the graph with the horizontal paths, in each path we will face the same situation: take the edge with value $1$ without observing the riskier option, or take the riskier option that has expected value of $1$. Therefore, in any horizontal path, an online algorithm can get a value of $1$ at most and as a result the competitive ratio of $ALG$ for this cover tends to $\frac{1}{2k}$ as $\varepsilon$ tends to zero. On the other hand, if we cover the graph with the $k$ vertical paths (completed with horizontal edges so that they start at $s$ and end at $t$), an algorithm that chooses a random path will get a value of $2$, and an algorithm that chooses the best path will even get a value of $k$.  
\end{example}

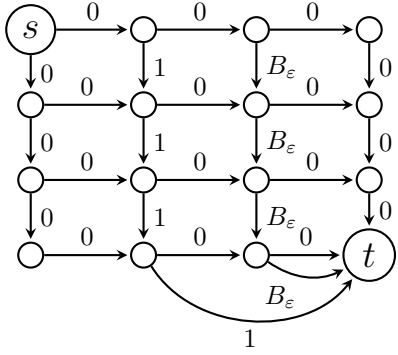
\begin{figure} [htb]
\centering
    
\begin{tikzpicture}[->,>=stealth,shorten >=1pt,auto,node distance=1.8cm,
                    thick,node_style/.style={circle,draw,font=\sffamily\Large\bfseries}]

    \foreach \x in {0,1,2,3} {
        \foreach \y in {0,1,2,3} {
            \ifnum\x=0
                \ifnum\y=0
                    \node[node_style] (\x-\y) at (\x*1.5,-\y) {$s$};
                \else
                    \node[node_style] (\x-\y) at (\x*1.5,-\y) {};
                \fi
            \else
                \ifnum\x=3
                    \ifnum\y=3
                        \node[node_style] (\x-\y) at (\x*1.5,-\y) {$t$};
                    \else
                        \node[node_style] (\x-\y) at (\x*1.5,-\y) {};
                    \fi
                \else
                    \node[node_style] (\x-\y) at (\x*1.5,-\y) {};
                \fi
            \fi
        }
    }
    
    \foreach \x in {0,3} {
        \foreach \y [count=\yi] in {0,1,2} {
            \draw (\x-\y) -- (\x-\yi) node[midway, right] {$0$};
        }
    }

    \foreach \x in {1} {
        \foreach \y [count=\yi] in {0,1,2} {
            \draw (\x-\y) -- (\x-\yi) node[midway, right] {$1$};
        }
    }

    \foreach \x in {2} {
        \foreach \y [count=\yi] in {0,1,2} {
            \draw (\x-\y) -- (\x-\yi) node[midway, right] {$B_\varepsilon$};
        }
    }
    
    \foreach \x [count=\xi] in {0,1,2} {
        \foreach \y in {0,1,2,3} {
            \draw (\x-\y) -- (\xi-\y) node[midway, above] {$0$};
        }
    }

    \draw (1-3) to[bend right=55] node[midway, below] {$1$} (3-3);
    \draw (2-3) to[bend right] node[midway, below left] {$B_\varepsilon$} (3-3);
    
\end{tikzpicture}
    \caption{Illustration for the instance in \Cref{ex:grid}. $B_\varepsilon$ denotes a realization of a variable that has value $1/\varepsilon$ with probability $\varepsilon$, and $0$ otherwise.}
    \label{fig:lowerbound_grid}
\end{figure}

\subsection{Disjoint Paths}
\label{sec:disjoint}

For the special case that the graph can be covered with $k$ node-disjoint $s,t$-paths (disjoint except for $s$ and $t$), we have the following result.

\begin{theorem}
    \label{thm:disjoint_paths}
    If the graph can be covered with $k$ node-disjoint paths, and edges are unlabeled, then there is an algorithm $ALG$ that satisfies $\E(ALG)\geq \frac{1}{k+1} \E(OPT)$.
\end{theorem}

For this, we need the following lemma, which is a modified version of \Cref{thm:width1_nolabel}, in which the benchmark takes the focal path with constant probability. We remark that this cannot be used to improve the bound of $1/2$, as in the worst-case instance, the probability that $OPT$ selects the focal path tends to $0$.

\begin{lemma}
    \label{lem:modified_width1}
    Consider a graph with a single focal path $P$, and an offline strategy $OFF$. Let $q$ be the probability that $OFF=P$. Then there is an algorithm $ALG$ that satisfies $\E(ALG)\geq \frac{1}{2-q}\E(OFF)$.
\end{lemma}

\begin{proof}
    Let $x_e$ be the probability that $OFF$ selects edge $e$, and enumerate the nodes in $P$ as $s=1,2,\dots,n=t$. At node $i$, the algorithm observes the realizations of edges that start at $i$ and tentatively chooses one of them with the same probability as $OFF$, conditional on the values of the edges that start at $i$. With probability
    \begin{align*}
        \alpha(i) = \frac{1/(2-q)}{1-\sum_{e\in E_{-i}} \frac{x_e}{2-q}},
    \end{align*}
    where $E_{-i}$ denotes the edges that skip $i$, the algorithm takes the tentative edge. With probability $1-\alpha(i)$, the algorithm stays in the path and goes to the next node $i+1$. If there is no tentative edge, the algorithm continues to the next node $i+1$. Since $OFF$ takes $P$ with probability $q$, $\sum_{e\in E_{-i}} x_e \leq 1-q$, so $\alpha(i)\in [0,1]$.

    Similar to the proof of \Cref{thm:width1_nolabel}, we can show inductively that $ALG$ takes each edge outside $P$ with probability exactly $x_e/(2-q)$, and every edge in $P$ with probability at least $x_e/(2-q)$. Assume this is true for edges that start in nodes $1,\dots,i-1$. For an edge $e$ that starts at node $i$, $ALG$ selects it with probability equal to (if it is the edge in $P$, with probability at least),
    \begin{align*}
        \left(
        1-\sum_{e'\in E_{-i}} \frac{x_{e'}}{2-1} \right) \cdot
        x_e \cdot \alpha(i) = \frac{x_e}{2-q}.
    \end{align*}
    As in the proof of \Cref{thm:width1_nolabel}, this, together with the definition of tentative edge, implies that $\E(ALG)\geq \frac{1}{2-q}\E(OFF)$.
\end{proof}

\begin{proof}[Proof of \Cref{thm:disjoint_paths}.]
Let $P_1,\dots P_k$ be the $k$ disjoint paths. Create $k$ graphs $G_1,\dots,G_k$ by taking in $G_i$ the nodes $V(P_i)$ and all edges that start at a node in $V(P_i)$, except for the edges that connect $s$ and $t$ directly. Assign to $G_1$ the edges that connect $s$ and $t$ directly. The sets $E(G_1),\dots,E(G_k)$ form a partition of $G$. Also, note that for every $s,t$-path $P$, and all $i=1,\dots,k$, either $P\subseteq E(G_i)$ or $P\cap E(G_i)=\emptyset$. Therefore, we can define the probability $f_i=\PP(OPT\subseteq E(G_i))$, and we have that $\sum_{i=1}^k f_i =1$. Now, let $OPT_i$ be an offline strategy in $G_i$ such that $OPT_i=OPT$ if $OPT\subseteq E(G_i)$, and $OPT=P_i$ otherwise. We have that
\begin{align*}
    \E(OPT)&= \sum_{i=1}^k \E(OPT\cdot \mathds{1}_{\{OPT\subseteq E(G_i)\}})\\
    &\leq \sum_{i=1}^k \E(OPT_i).
\end{align*}

If we apply \Cref{lem:modified_width1} to $G_i$, with $OFF=OPT_i$, we obtain an algorithm $ALG_i$ such that $\E(ALG_i)\geq \frac{1}{1+f_i}\cdot \E(OPT_i)$. So, if we take $ALG$ to be $ALG_i$ for $i$ that maximizes the expected value, we obtain that
\begin{align*}
    \E(ALG)&\geq \max_{1\leq i\leq k} \frac{\E(OPT_i)}{1+f_i}.
\end{align*}
Therefore,
\begin{align*}
    \E(OPT)\leq \sum_{i=1}^k \E(OPT_i) \leq \sum_{i=1}^k (1+f_i)\cdot \max_{1\leq j\leq k} \frac{\E(OPT_j)}{1+f_j}\leq \E(ALG)\cdot \sum_{i=1}^k (1+f_i) = (k+1)\cdot \E(ALG).
\end{align*}
This concludes the proof of the theorem.  
\end{proof}

\subsubsection{Upper bound}
\label{sec:upperbounds}

We present an upper bound of $1/(k+1)$ on an instance with $k$ disjoint paths, which shows the tightness of \Cref{thm:disjoint_paths}. 

\begin{example}[Upper bound of $1/(k+1)$ for disjoint paths] \label{ex:kplus1} Consider the instance depicted in \Cref{fig:lowerbound_kplus1}, where apart from $s$ and $t$, we have $k$ nodes numbered $1,\dots, k$. For each node $i=1,\dots, k$, there is an edge from $s$ to $i$ with a value of $0$, and an edge from $i$ to $t$ with a value of $1/\varepsilon$ with probability $\varepsilon$ and $0$ otherwise. All values are independent. We also have one extra edge that connects $s$ and $t$ and has a value of $1$. In this instance, $OPT$ takes $1/\varepsilon$ if any is realized, and the $1$ if none is realized. Thus, $\E(OPT)=(1/\varepsilon)\cdot (1-(1-\varepsilon)^k) + 1\cdot (1-\varepsilon)^k=k+1-O_k(\varepsilon)$, while any online algorithm can get at most a value of $1$.
    
\end{example}

\begin{figure}[htb]
    \centering

\begin{tikzpicture}[->,>=stealth,shorten >=1pt,auto,node distance=1.6cm,
                    thick,node_style/.style={circle,draw,font=\sffamily\Large\bfseries}]

    \node[node_style] (s) {$s$};
    \node[node_style, right = of s] (i1) {$1$};
    \node[node_style, below of=i1, yshift=10] (i3) {$2$};
    \node[node_style, below of=i3] (i4) {$k$};
    \node[node_style, right = of i1] (t) {$t$};
    \node[below of=i3, yshift=25pt] (p1) {$\vdots$};

    \draw[->] (s) -- (i1) node[midway, above] {0};
    \draw[->] (s) to[bend right=10] node[pos=0.55, above=2pt] {0} (i3);
    \draw[->] (s) to[bend right=20]  node[midway, right=3pt] {0} (i4);
    \draw[->] (i1) -- (t) node[midway, above] {$B_\varepsilon$};
    \draw[->] (i3) to[bend right=10] node[pos=0.45, above=2pt] {$B_\varepsilon$} (t);
    \draw[->] (i4) to[bend right=20] node[midway, left=1pt] {$B_\varepsilon$} (t);
    
    \draw[->] (s) to[bend left=45] node[above] {$1$} (t);
\end{tikzpicture}
    
    \caption{Instance of \Cref{ex:kplus1} that attains the upper bound of $1/(k+1)$ for graphs of width $k$. Each $B_\varepsilon$ denotes an independent realization of a distribution that takes the value $1/\varepsilon$ with probability $\varepsilon$ and $0$ otherwise.}
    \label{fig:lowerbound_kplus1}
\end{figure}
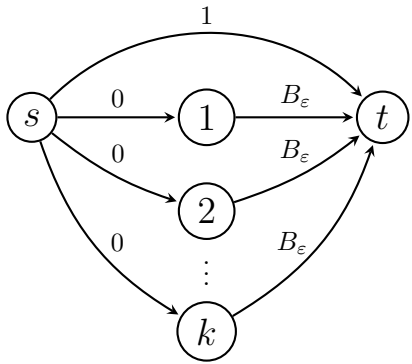

\section{Discussion}
In this paper, we present a powerful graph-theoretic framework for prophet inequalities. We obtain a $\frac{1}{k(d+2)}$-competitive algorithm for a graph that can be covered by $k$ paths and each edge has at most $d$ labels. A natural open question resulting from our work is whether there is always a path cover allowing for a $\frac{1}{k+1}$-competitive algorithm in the case $d=0$. If there is such a cover, is it possible to compute it in polynomial time? Our framework sets the stage for many exciting extensions, for example:
\begin{itemize}
\item Improving bounds by restricting the distributions: For the classic prophet inequality, substantially better algorithms are available when the values are IID \cite{CFHOV21}. Does something similar happen in our framework when the value of every edge is either $0$ or drawn from the same distribution IID? Is there an algorithm with a better competitive ratio for this case? 
\item Handling a DM that only gets a few samples from each distribution: A recent line of work considers the case in which the DM does not know the distributions that the values are sampled from but is given instead some samples from each distribution, \cite{CCES23,AKW14,CDFF22,KNR22}. It is interesting to consider a similar question in our graph theoretic framework and figure out how the competitive ratio achievable by the algorithm changes when we get some number of samples for the value of each edge.
\end{itemize}

\section*{Acknowledgments}

This material is based upon work partially supported by the National Science Foundation under Grant No. DMS-1928930, while the authors were in residence at the Mathematical Sciences Research Institute in Berkeley, California, during the fall semester of 2023. This work was supported in part by BSF grant 2018206, ISF grant 2167/19, Simons Foundation Collaboration on the Theory of Algorithmic Fairness, Simons Foundation investigators award 689988, Centro de Modelamiento Matemático (CMM) BASAL fund FB210005 for center of excellence from ANID-Chile, and Koret Program for Scholars from Israel while S. Oren was visiting Stanford.

 \bibliographystyle{acm}

 \bibliography{references}

\end{document}